\documentclass[letterpaper, 10 pt, conference]{ieeeconf}  

\IEEEoverridecommandlockouts                              

\usepackage{fullpage,marginnote,caption}
\usepackage{amsmath}
\usepackage[utf8]{inputenc} 
\usepackage[T1]{fontenc}    
\usepackage{hyperref}       
\usepackage{url}            
\usepackage{booktabs}       
\usepackage{amsfonts}       
\usepackage{nicefrac}       
\usepackage{microtype}      
\usepackage{empheq}
\usepackage{amsfonts}
\usepackage{amssymb}
\usepackage{graphicx}
\usepackage{fancyhdr,color}
\usepackage{mathtools}
\usepackage{biblatex}
\usepackage{bbm}
\usepackage{bm}
\usepackage{mathabx}
\usepackage{algorithm}
\usepackage{algpseudocode}
\usepackage{enumerate} 
\usepackage[dvipsnames]{xcolor}
\newtheorem{theorem}{Theorem}
\newtheorem{remark}{Remark}

\newtheorem{definition}{Definition}
\newtheorem{assumption}{Assumption}

\newtheorem{corollary}{Corollary}

\newtheorem{lemma}{Lemma}
\newcommand{\RomanNumeralCaps}[1]
    {\MakeUppercase{\romannumeral #1}}
\graphicspath{ {figures/} }
\usepackage{balance}
\title{\LARGE \bf
Output Feedback Stochastic MPC with Hard Input Constraints}

\author{Eunhyek Joa, Monimoy Bujarbaruah, and Francesco Borrelli} 

\begin{document}

\noindent

\maketitle
\thispagestyle{empty}
\pagestyle{empty}

\begin{abstract}
We present an output feedback stochastic model predictive controller (SMPC) for constrained linear time-invariant systems. 
The system is perturbed  by additive Gaussian disturbances on state and  additive Gaussian measurement noise on output. A Kalman filter is used for state estimation and a SMPC is designed
to satisfy chance constraints on states and hard constraints on actuator inputs. 
The proposed SMPC constructs bounded sets for the state evolution and a tube based constraint tightening strategy where the tightened constraints are time-invariant.



We prove that the proposed SMPC can guarantee an infeasibility rate below a user-specified tolerance.
We numerically compare our method with a classical output feedback SMPC with simulation results which highlight the efficacy of the proposed algorithm. 
\end{abstract}

\section{Introduction}
Over the last few decades, Model Predictive Control (MPC) has been successfully implemented for constrained control of dynamical systems \cite{hrovat2012development, mayne2014model, morari1999model}. 
A key challenge in MPC design is handling the presence of disturbances in the system model and noisy output feedback. Robust MPC (RMPC) is a well studied field for MPC design which tackles this challenge\cite{bemporad1999robust, mayne2014model,  mayne2006robust, mayne2015robust}.
In RMPC, the controller is designed to satisfy constraints for all possible realizations of bounded disturbances \cite{chisci2001systems,langson2004robust, mayne2005robusto} and measurement noises \cite{brunke2021rlo, mayne2006robust}.
However, due to the inherent robustness to all realizations of uncertainty and noise, RMPC is often rendered conservative \cite{lorenzen2016constraint, rosolia2018stochastic}.
For this reason, Stochastic MPC (SMPC) has been widely studied for systems which are not safety critical, and thus, state constraint violations can be allowed to a tolerable extent~\cite{hewing2020recursively,lorenzen2016constraint,mesbah2016stochastic}. SMPC provides a systematic way to consider the trade-off between satisfying imposed constraints with high probability and reducing conservative controller behavior. 

Although SMPC has been widely studied \cite{kouvaritakis2004recent,mesbah2016stochastic}, most of these approaches assume perfect state measurement, which is often not the case in most real-world systems. Output feedback SMPC has been studied in \cite{cannon2012stochastic, farina2015approach, hokayem2012stochastic}, where the approach is to design a state estimator, analyze statistics of its estimation errors, and utilize the error bounds to tighten the imposed constraints accordingly. The systems considered in \cite{cannon2012stochastic, farina2015approach, hokayem2012stochastic} are linear time-invariant (LTI) systems perturbed by bounded additive stochastic disturbances on state and additive stochastic noises on output. In \cite{cannon2012stochastic}, the method relies on knowing the supports of disturbances and noises.
However, in practice, it may not be straightforward to estimate the support \cite{ bujarbaruah2020learning_ball_in_cup}, and only the statistics of disturbances and noises might be available \cite{mesbah2016stochastic}.
In \cite{farina2015approach}, the supports of disturbance and noise are not bounded. However, inputs are not imposed on hard input constraints. 
In \cite{hokayem2012stochastic}, hard input constraints are considered without considering state constraints.

This paper focuses on  output feedback SMPC, which is designed to satisfy chance state constraints and hard input constraints, for LTI systems perturbed by additive Gaussian disturbances on state and additive measurement noises on output. 
As in \cite{hokayem2012stochastic, ridderhof2020chance}, we use a Kalman filter for state estimation. We analyze stochastic properties of two random vectors: the estimation error and the control error. 
Exploiting these properties, we construct time-invariant bounded sets for these errors dynamics so that the errors are guaranteed to remain in the bounded sets with  a user-specified bound on the probability.
We refer such sets as a \textit{Uniformly Bounded Confidence Sets}, which will be formally defined in Section~\ref{sec: UBCS}.
These sets are employed to tighten the imposed state and input constraints and design an output feedback SMPC, which guarantees recursive feasibility in probability with a user-specified bound.

Our key contributions can be summarized as follows:
\begin{itemize}
    \item A systematic way to construct \textit{Uniformly Bounded Confidence Sets} for the estimation and the control errors when system states are estimated from the system outputs using a Kalman filter.
    \item Guaranteeing a user-specified probability of success in the task using a shrinking tube based MPC \cite{chisci2001systems}, where the constraint tightening strategy uses the \textit{Uniformly Bounded Confidence Sets}.
\end{itemize}
Instead of constructing time-varying tube to tighten the constraints \cite{cannon2012stochastic, farina2015approach}, we construct rigid \textit{Uniformly Bounded Confidence Sets}.
To avoid infeasibility of the MPC, existing works assume that the supports of disturbance and noise are known \cite{cannon2012stochastic}, the input can be unbounded \cite{farina2015approach}, or the state is unconstrained \cite{hokayem2012stochastic}. 
In this work, instead of adopting additional assumptions, we tighten the constraints to guarantee recursive feasibility with a user-specified probability bound. 

\textit{Notation:} Throughout the paper we use the following notation. The positive (semi)definite matrix \(P\) is denoted as (\(P\succcurlyeq0\)) \(P\succ0\).
The Minkowski sum of two sets is denoted as \(\mathcal{X} \oplus \mathcal{Y} = \{x+y: x \in \mathcal{X}, y \in \mathcal{Y}\}\).
The Pontryagin difference between two sets is defined as \(\mathcal{X} \ominus \mathcal{Y} = \{x\in\mathcal{X}: x+y \in \mathcal{X}, \forall y \in \mathcal{Y}\}\).
The m-th row vector of a matrix \(H\) is denoted as \([H]_m\), similarly the m-th component of a column vector \(h\) is \([h]_m\). 
\(\mathbb{P}(\mathcal{A})\) is the probability of event \(\mathcal{A}\), and \(\mathbb{E}[\cdot]\) is the expectation of its argument. 
The notation \(x_{l:m}\) means the sequence of the variable \(x\) from time step \(l\) to time step \(m\). 

\section{Problem setup}
We consider an uncertain LTI system perturbed by a state additive disturbance.
We obtain noisy output feedback from the system during control synthesis. 
\subsection{System Dynamics, Measurement Model}
The system dynamics and measurement model are given as follows:
\begin{equation} \label{eq:system}
\begin{split}
    & x_{k+1}  = A x_k  + B u_k  + w_k ,\, w_k  \sim \mathcal{N}(0,Q), \\
    & y_k  = C x_k  + v_k ,\, v_k  \sim \mathcal{N}(0,R), \\
    & x_0  \sim \mathcal{N}(\mu_0 ,\Sigma_0 ),
\end{split}
\end{equation}
where \(x_k  \in \mathbb{R}^{n_x}\) is the state, \(u_k  \in \mathbb{R}^{n_u}\) is the input, and \(y_k  \in \mathbb{R}^{n_y}\) is the measurement at time step $k$. System matrices \(A,B\) and the measurement model matrix \(C\) are known. We assume that the system is controllable and observable \cite{ogata2010modern}.
At each time step \(k\), the system is affected by an independently and identically distributed (i.i.d.) Gaussian random disturbance \(w_k  \in \mathbb{R}^{n_x}\) and measurement noise \(v_k  \in \mathbb{R}^{n_y}\), both with zero mean and known covariance matrices $Q,R$, respectively. 
Moreover, the initial state of the system \(x_0 \) is a Gaussian random vector with a known mean $\mu_0$ and covariance matrix $\Sigma_0$.

\subsection{Estimator}
We estimate the state of the system for control synthesis using the following estimator:
\begin{equation} \label{eq: general estimator}
\begin{split}
    & \hat{x}_k = f_k(\hat{x}_{k-1}, y_{k}), 
\end{split}
\end{equation}
where $\hat{x}_k$ is an estimated state at time step $k$.
There are multiple choices of the estimator design, such as Luenberger observer \cite{cannon2012stochastic, farina2015approach, mayne2006robust} with a static gain or a Kalman filter \cite{ridderhof2020chance}. In this work we use a Kalman filter design. 
Since system (\ref{eq:system}) is LTI and the additive disturbance and the measurement noise are Gaussian distributed, the Kalman filter is the optimal estimator in the minimum mean-square-error sense~\cite[Chap.E.3]{bertsekas2012dynamic} for the unconstrained system.
Using the Kalman filter, the estimated state can be obtained as follows \cite{simon2006optimal}:
\begin{equation} \label{eq:kalman filter}
\begin{split}
    & \hat{x}_{k+1}^{-} = A \hat{x}_{k} + B u_{k} , \\
    & \hat{x}_{k+1} = \hat{x}_{k+1}^{-} + L_{k+1} (y_{k+1}  - C \hat{x}_{k+1}^{-}), 
\end{split}
\end{equation}
where $\hat{x}_{k}^{-}$ is \textit{a priori} state estimate at time step $k$, and \(L_{k+1}\) is the Kalman filter gain
is derived as follows \cite{simon2006optimal}:
\begin{equation} \label{eq:kalman filter_cov_gain}
\begin{split}
    & P_{k+1}^- = A P_k A^\top + Q, \\
    & P_{k+1}  = (I - L_{k+1} C) P_{k+1}^-, \\
    & L_{k+1} = P_{k+1}^- C^\top (CP_{k+1}^-C^\top + R)^{-1}, 
\end{split}
\end{equation}
where \(P_{k} :=\mathbb{E}[(x_k - \hat{x}_{k} ) (x_k - \hat{x}_{k} )^\top|y_{0:k}]\) and \(P_{k}^-:=\mathbb{E}[(x_k - \hat{x}_{k}^-) (x_k - \hat{x}_{k}^-)^\top|y_{0:k}]\).
The Kalman filter is initialized based on the known statistics of the initial condition in (\ref{eq:system}):
\begin{equation} \label{eq:kalman filter_initialization}
\begin{split}
    & \hat{x}_0^{-} = \mu_0, \,\,  P_0^{-}=\Sigma_0.
\end{split}
\end{equation}

\subsection{Constraints}
System (\ref{eq:system}) is subject to the following constraints at all time steps:
\begin{equation} \label{eq:constraints}
\begin{split}
    & \qquad \, \, \mathbb{P}(x_{k} \in \mathcal{X}) \geq 1-p_x~\forall k \geq 0 ,\\
    & \qquad \, \, u_k \in \mathcal{U}, ~\forall k \geq 0, 
\end{split}
\end{equation}
where \(\mathcal{X}\) is a convex polyhedron, \(\mathcal{U}\) is a polytope, and $p_x \in (0,1)$ is the allowable violation probability of the state constraints $\mathcal{X}$ by the system~(\ref{eq:system}) in closed-loop with an output-feedback controller.

\subsection{Optimal Control Problem}
We want to compute solutions to the following finite time optimal control problem:
\begin{equation} \label{eq:ftocp}
\begin{split}
    & V (\mu_0, \mathcal{F}) = \\
    & \min_{U(\cdot)} \,\, \sum_{k=0}^{T-1} \ell(\hat{x}_k , \pi_k (\hat{x}_k)) \\
    & \,\,\, \textnormal{s.t.,} \,\,\,  x_{k+1}  = A x_k  + B \pi_k (\hat{x}_k) + w_k ,\\
    & \qquad \, \, y_k  = C x_k  + v_k ,  \\
    & \qquad \, \, \hat{x}_{k+1} = f_k(\hat{x}_{k}, y_{k+1}), \\
    & \qquad \, \, \hat{x}_{0} = \mu_0 + L_0 (y_0 - C \mu_0), \\
    & \qquad \, \, \mathbb{P}(x_{k}  \in \mathcal{X}) \geq 1-p_x, \, \pi_k (\hat{x}_k) \in \mathcal{U},\\
    & \qquad \, \, x_0  \sim \mathcal{N}(\mu_0 ,\Sigma_0 ), \\ 
    & \qquad \, \, w_k  \sim \mathcal{N}(0,Q), \,\, v_k  \sim \mathcal{N}(0,R),\\
    & \qquad \, \, k=0,...,T-1, 
\end{split}
\end{equation}
where $T\gg0$ is the task horizon and $\mathcal{F}$ is a set of given time-varying state estimator functions $\{f_0(\cdot,\cdot),f_1 (\cdot,\cdot),...,f_{T-1} (\cdot,\cdot)\}$.
The optimal control problem (\ref{eq:ftocp}) minimizes the sum of the positive definite state cost $\ell(\cdot,\cdot)$ evaluated for the estimated state trajectory.
We point out that as the system (\ref{eq:system}) is uncertain and provides noisy output feedback $y_{0:k}$, the optimal control problem (\ref{eq:ftocp}) consists of finding state feedback policies \(U (\cdot) = \{\pi_0(\cdot),\pi_1(\cdot),...,\pi_{T-1}(\cdot)\} \).
\begin{remark}
The existing works \cite{cannon2012stochastic,farina2015approach,hewing2020recursively} have chance constraints on inputs. We consider hard input constraints in \eqref{eq:ftocp}, which are more common in practical applications.  
\end{remark}

\subsection{Solution Approach to (\ref{eq:ftocp})}
\label{sub: Solution Approach to ftocp}
There are two main challenges in solving \eqref{eq:ftocp}, namely: 
\begin{enumerate}[(C1)]
    \item Optimizing over control policies \( \{\pi_0(\cdot),\pi_1 (\cdot),...\} \) involves an infinite dimensional optimization.
    \item Solving the problem \eqref{eq:ftocp} for $T \gg 0$ is computationally demanding.
\end{enumerate}
We address (C1) by approximating control policy to an affine estimated state feedback policy with a fixed control gain as in \cite{cannon2012stochastic, farina2015approach, hewing2020recursively}. That is,
\begin{equation} \label{eq:affine control}
\begin{split}
    \pi_{i|k}(x) = c_{i|k} + K(x - \bar{x}_{i|k}),
\end{split}
\end{equation}
where $\bar{x} _{i|k}$ is nominal state, $c _{i|k}$ is an auxiliary input, and $\pi_{i|k}(\cdot)$ is the control policy at predicted time step $k+i$. \(K\) is a control gain chosen such that \(A_\mathrm{cl} := A+BK\) is Schur.
We address (C2) by solving a simpler constrained optimal control problem with prediction horizon $N \ll T$ in a receding horizon fashion.
Specifically, we attempt to design an MPC controller by solving:
\begin{equation} \label{eq:MPC}
\begin{split}
    & \bar{V}_{k \rightarrow k+N} (\hat{x}_k) = \\
    & \min_{\substack{c _{0|k},...c _{N-1|k}}} \sum_{i=0}^{N-1} \ell(\bar{x} _{i|k}, c _{i|k}) + O (\bar{x} _{N|k})\\
    & \qquad \textnormal{s.t.,} \quad  \quad \bar{x} _{i+1|k} = A \bar{x} _{i|k} + B c _{i|k}, \\
    & \quad \qquad \qquad \, \, \bar{x} _{0|k} = \hat{x}_k,\\
    & \quad \qquad \qquad \, \,  \bar{x} _{i|k} \in \Bar{\mathcal{X}}_{\mathrm{RF}, i}, \\
    & \quad \qquad \qquad \, \,  c _{i|k} \in \Bar{\mathcal{U}}_{\mathrm{RF},i},\\
    & \quad \qquad \qquad \, \, \bar{x} _{N|k} \in \Bar{\mathcal{X}} _{f,\mathrm{RF}},\\
    & \quad \qquad \qquad \, \, i=0,1,...,N-1,
\end{split}
\end{equation}
where $\hat{x}_k$ is obtained from the Kalman filter \eqref{eq:kalman filter} and $O(\cdot)$ is a terminal cost.
As $\pi_{i|k}(\bar{x}_{i|k}) = c_{i|k}$, the sequence $\{c _{0|k},...c _{N-1|k}\}$ is optimized. 
In \eqref{eq:MPC}, state constraints $\Bar{\mathcal{X}}_{\mathrm{RF}, i}$ and input constraints $\Bar{\mathcal{U}}_{\mathrm{RF},i}$ are obtained by tightening the original constraints \eqref{eq:constraints} with \textit{Uniformly Bounded Confidence Sets}. The construction of these sets is  discussed in detail in Section~\RomanNumeralCaps{3}.
The choice of the terminal set $\bar{\mathcal{X}}_{f,\mathrm{RF}}$ is explained in Section \RomanNumeralCaps{4}.
After solving \eqref{eq:MPC}, we apply
\begin{equation} \label{eq:MPC policy}
    \pi_k(\hat{x}_k) = \pi_{0|k}(\hat{x}_k) = c^\star_{0|k} 
\end{equation}
to system \eqref{eq:system} in closed-loop. 

\begin{remark}
In  \eqref{eq:MPC} we find an optimal input sequence $c^\star_{i|k}$ based on the nominal system initialized with the estimated state. 
Thus, to satisfy \eqref{eq:constraints}, we need to consider two errors predicted along the control horizon: (1) error between the actual and the estimated states, (2) error between the estimated and the nominal states.
The same approach has been followed in \cite{mayne2005robusto, cannon2012stochastic, farina2015approach, hokayem2012stochastic}.
In Section \ref{sub:Uniformly Bounded Confidece Set} we propose a novel approach which constructs a time-invariant polytope set for each error type such that for all time step $k \in \{0,...,T-1\}$, each error is in the corresponding set with probability no smaller than a user-defined value.
In Section \RomanNumeralCaps{4} we tighten the constraints \eqref{eq:constraints} with these time-invariant sets in a \textit{shrinking tube} way \cite{chisci2001systems} so that imposing $\Bar{\mathcal{X}}_{\mathrm{RF}, i}$ on the nominal state $\bar{x}_{i|k}$ is sufficient to satisfy the constraints \eqref{eq:constraints}.  
For this reason the MPC formulation in \eqref{eq:MPC} and the policy in \eqref{eq:MPC policy} do not include predicted errors and predicted input policies along the horizon.
\end{remark}

Since the supports of the additive disturbance \(w_k \) and measurement noise \(v_k \) are unbounded, and the control inputs \(u_k \) are bounded, this can result in the loss of feasibility of (\ref{eq:MPC}).
Accordingly, we consider the following design specifications.
\begin{enumerate}[(D1)]
    \item (Safety) If the MPC \eqref{eq:MPC} is feasible at time step $k$ and the closed-loop control law \eqref{eq:MPC policy} is applied to the system \eqref{eq:system}, then the satisfaction of the constraints in \eqref{eq:constraints} is guaranteed at time $k$.
    \item (Regulation of the loss of feasibility) If the MPC \eqref{eq:MPC} is feasible at time step $k$ and the closed-loop control law \eqref{eq:MPC policy} is applied to the system \eqref{eq:system}, the MPC \eqref{eq:MPC} at time step $k+1$ is feasible with probability no smaller than a user specified bound $1-p_f \in (0,1)$.
\end{enumerate}
Note that (D2) implies that the MPC \eqref{eq:MPC} is infeasible with probability less than $p_f$ at each time step, if \eqref{eq:MPC} is feasible at the previous time step. Therefore for the given task with \(T\) horizon, the probability of success of the task is \((1-p_f)^{T-1}\), if \eqref{eq:MPC} is feasible at $k=0$.

\section{Constructing Uniformly Bounded Confidence Sets} \label{sec: UBCS}
In this section, we define the term \textit{Uniformly Bounded Confidence Set} and propose a systematic way to construct \textit{Uniformly Bounded Confidence Sets} for the estimation and the control errors when actual states are estimated using the Kalman filter. This is our first contribution of the paper. 

We define a \textit{Uniformly Bounded Confidence Sets} of the given random variable as follows:
\begin{definition} \label{def:uniform confidence set}
(Uniformly Bounded Confidence Set)
A set \(\mathcal{E}_{p}^{r}\) is a uniformly bounded confidence set of probability level \(p\) for a random variable \(r_k\) if
\begin{equation} \label{eq:uniform confidence set}
\begin{split}
    & \mathbb{P}(r_k \in \mathcal{E}_{p}^{r}) \geq p, \, \forall k \geq 0. 
\end{split}
\end{equation}
\end{definition}
In next two subsections \ref{sub:Estimation error} and \ref{sub:Additive disturbance in control error dynamics}, we introduce two error random variables for which we want to construct \textit{Uniformly Bounded Confidence Sets}. 
\subsection{Estimation Error} \label{sub:Estimation error}
In this subsection, we introduce the notion of estimation error, which is the error between the actual and estimate states.
We derive the time evolution dynamics of the estimation error based on the Kalman filter (\ref{eq:kalman filter}) and the system (\ref{eq:system}) as follows:
\begin{equation} \label{eq:estimation error dynamics}
\begin{split}
    & e_k  := x_k  - \hat{x}_{k}, \\
    & e_{k+1}  = (I - L_{k+1} C) A e_k  \\ 
    & \qquad \quad + (I - L_{k+1} C) w_k  - L_{k+1} v_{k+1}. 
\end{split}
\end{equation}
Note that \(\mathbb{E}[e_k  \hat{x}_{k}^{\top}|y_{0:k}]=0\) and \(\mathbb{E}[e_k|y_{0:k}]=0\) following~\cite[Cor.E.3.2]{bertsekas2012dynamic}.
Thus, the covariance matrix of the estimation error $e_k$ is  $\mathbb{E}[e_k e_k^\top|y_{0:k}] = \mathbb{E}[(x_k - \hat{x}_{k} ) (x_k - \hat{x}_{k} )^\top|y_{0:k}] =  P_{k}$. Though the control inputs \eqref{eq:MPC policy} will be nonlinear functions of the estimated state \cite{maciejowski2002predictive}, the effect of the control inputs is cancelled out in the dynamics of \(e_k \), and \(e_k \) is a linear sum of the independent normal random vectors as shown in (\ref{eq:estimation error dynamics}).
Therefore, the conditional probability distribution of the estimation error \(e_k \) given the measurements up to time step \(k\), \(y_{0:k} \) is given by:
\begin{equation} \label{eq:cond dist plus of e}
\begin{split}
    & p(e_k |y_{0:k} ) \sim \mathcal{N}(0, P_k). 
\end{split}
\end{equation}

\subsection{Additive Disturbance in Control Error Dynamics} \label{sub:Additive disturbance in control error dynamics}
In this subsection, we introduce the notion of control error (the error between the estimated and the nominal states), derive its dynamics, and analyze the stochastic property of the additive disturbance on the control error.

First, we derive the dynamics of the estimated state $\hat{x}_{k}$ from the Kalman filter (\ref{eq:kalman filter}) as follows:
\begin{equation} \label{eq:estimated state dynamics}
\begin{split}
    & n_k  := L_{k+1} C A e_k  + L_{k+1} C w_k  + L_{k+1} v_{k+1} , \\
    & \hat{x}_{k+1}  = A\hat{x}_k + B u_k + n_k , 
\end{split}
\end{equation}
where \(n_k\) can be interpreted as an additive disturbance to the estimator dynamics.
As we will present later, \(n_k\) is also the additive disturbance on the control error.

Second, we analyze the stochastic property of the additive disturbance $n_k$.
\(n_k\) is a linear combination of three independent Gaussian random variables:~\(e_k\),~\(w_k\), and \(v_k\)  as shown in (\ref{eq:estimated state dynamics}).
Note that these are independent because \(w_k\) and \(v_k\) are i.i.d. Gaussian random vectors. For these reasons, \(n_k\) is a Gaussian random variable itself. The mean of \(n_k \) is zero, i.e., \(\mathbb{E}[n_k |y_{0:k} ]=0\), because \(\mathbb{E}[e_k |y_{0:k} ]=\mathbb{E}[w_k ]=\mathbb{E}[v_k ]=0\). 
The covariance matrix of \(n_k\) can be derived as follows:
\begin{equation} \label{eq:noise covariance dynamics}
\begin{split}
    & \Phi_k  := \mathbb{E}[n_k  n_k^{\top} |y_{0:k}], \\
    &  \Phi_k  = L_{k+1} (C (A P_{k}  A^\top + Q) C^\top + R) L_{k+1}^\top. 
\end{split}
\end{equation}
Thus, \(p(n_k |y_{0:k} ) \sim \mathcal{N}(0,\Phi_k )\).
Note that the mean and the covariance of $p(n_k |y_{0:k})$ are not a function of $y_{0:k}$.

Third, based on the dynamics of the nominal system in \eqref{eq:MPC}, the dynamics of the estimated state in \eqref{eq:estimated state dynamics}, the affine control policy in \eqref{eq:affine control}, and the conditional distribution of $n_k$ in \eqref{eq:noise covariance dynamics}, we calculate the estimated state $\hat{x}_{i|k}$ at predicted time step $k+i$ for $i\geq0$ as follows:
\begin{subequations} \label{eq:predicted estimated state dynamics}
\begin{align}
    & \hat{x}_{0|k} = \hat{x}_{k}, \label{subeq: estimated state initialization} \\
    & \hat{x}_{i+1|k}  = A\hat{x}_{i|k} + B\pi_{i|k}(\hat{x}_{i|k}) + n_{i|k},  \label{subeq: estimated state dynamics}\\
    & n_{i|k} \sim p(n_{k+i}|y_{0:k+i}), \label{subeq: estimated state disturbance}
\end{align}
\end{subequations}
where $\hat{x}_{0|k}$ is initialized with $\hat{x}_{k}$ as shown in \eqref{subeq: estimated state initialization}.
For each predicted time step $k+i$, we draw $n_{i|k}$ from $p(n_k+i|y_{0:k+i})$ as shown in \eqref{subeq: estimated state disturbance}.
Note that we can utilize $p(n_k+i|y_{0:k+i})$ without information of future measurements $y_{k+1:k+i}$ because the mean and covariance of Gaussian distribution $p(n_{k+i}|y_{0:k+i})$ are not functions of $y_{0:k+i}$.

Finally, we introduce the control error and explain why the stochastic property of the additive disturbance $n_k$ matters. 
We define the control error and derive its dynamics using the nominal state dynamics in \eqref{eq:MPC} and the predicted estimated state dynamics in \eqref{eq:predicted estimated state dynamics} as follows:
\begin{equation} \label{eq:control error dynamics}
\begin{split}
    & \Delta_{i|k}  := \hat{x}_{i|k}  - \bar{x}_{i|k} , \\
    & \Delta_{i+1|k}  = A_\mathrm{cl}\Delta_{i|k}  + n_{i|k}, \,\, n_{i|k} \sim p(n_{k+i}|y_{0:k+i}). \\ 
\end{split}
\end{equation}
As we initialize both $\bar{x}_{0|k}$ and $\hat{x}_{0|k}$ with the state estimate $\hat{x}_k$ in \eqref{eq:MPC} and \eqref{eq:predicted estimated state dynamics}, $\Delta_{0|k}= 0$ and $\Delta_{i+1|k} = {\scriptstyle\sum}_{q=0}^{i} A_\mathrm{cl}^{i-q} n_{q|k}$ based on \eqref{eq:control error dynamics}. 
As the predicted control error $\Delta_{i+1|k}$ is a sum of the additive disturbance of \eqref{eq:control error dynamics}, we analyzed the stochastic property of the additive disturbance $n_k$ and we will construct a \textit{Uniformly Bounded Confidence Set} of $n_k$ in Section~\RomanNumeralCaps{3}.D. 
In Section~\RomanNumeralCaps{4}, we will tighten constraints using this \textit{Uniformly Bounded Confidence Set} of $n_k$ so that we can bound predicted control errors $\Delta_{i|k}$ around the predicted nominal state $\bar{x}_{i|k}$ with no smaller than a user-defined probability.

\subsection{Uniform upper bounds to the covariance matrices}
\label{sub:Uniform upper bound}
We derive \(P\) and \(\Phi\) such that \(P_k \preccurlyeq P\) and \(\Phi_k \preccurlyeq \Phi\) for all $k \in \{0,1,...,T-1\}$.
We refer \(P\) and \(\Phi\) to the uniform upper bound to $P_k$ and $\Phi_k$ for all $k \in \{0,1,...,T-1\}$, respectively.  
We present two distinct choices to derive \(P\) and \(\Phi\).

\subsubsection*{Solving L\"{o}wner-John ellipsoid problem \cite{boyd2004convex}} First, we can numerically calculate both uniform upper bounds \(P\) and \(\Phi\) by solving the L\"{o}wner-John ellipsoid problem \cite{boyd2004convex}.
We refer these bounds as \(P_\mathrm{LJ}\) and \(\Phi_\mathrm{LJ}\).
As the horizon of the given task \eqref{eq:ftocp} is finite, the number of $P_k $ and $\Phi_k, \, \forall k \in \{0,...,T-1\}$ is finite.
Thus, finding an uniform upper bound to $P_{0:T-1} $ or $\Phi_{0:T-1}$ can be reformulated as a e L\"{o}wner-John ellipsoid problem \footnote{The formulation of the L\"{o}wner-John ellipsoid problem is presented in the Appendix.} \cite{boyd2004convex}: finding minimal volume ellipsoid $x^\top {P}^{-1}_\mathrm{LJ} x \leq 1$ ($x^\top \Phi^{-1}_\mathrm{LJ} x \leq 1$) which contains all given ellipsoids $x^\top {P_k }^{-1} x \leq 1 \, (x^\top \Phi_k^{-1} x \leq 1) \, \forall k \in \{0,...,T-1\}$.
If some covariance matrices are positive semidefinite, then one can add $\epsilon \rm{I}$ to those matrices.

\subsubsection*{Mathematical derivation} 
In additon to the above approach,  uniform upper bounds \(P\) and \(\Phi\) can be also derived explicitly under the  assumptions introduced next. We refer to these bounds as \(P_\mathrm{MD}\) and \(\Phi_\mathrm{MD}\), respectively.
\begin{assumption} \label{assum: initial condition}
We assume that the covariance matrix of the initial state in (\ref{eq:system}), which is also initial covariance matrix of the Kalman filter in (\ref{eq:kalman filter_initialization}), satisfies the following condition:
\begin{equation} \label{eq:initial covariance bnd}
\Sigma_0 = P_0^{-} \preccurlyeq P_{\infty}
\end{equation}
where $P_{\infty}$ is a solution of the following equation,
\begin{equation} \label{eq:steady state covariance}
P_{\infty} = A (P_{\infty}-P_{\infty}C^T(CP_{\infty}C^T+R)^{-1}CP_{\infty})A^T +Q.
\end{equation}
\end{assumption}
\begin{assumption} \label{assum: system matrix}
The system matrix $A$ is nonsingular.
\end{assumption}

Under assumptions \ref{assum: initial condition}-\ref{assum: system matrix}, we derive a uniform bound \(P_\mathrm{MD}\) to \(P_k \) for any $k \geq 0$.
To do that, we obtain a uniform upper bound to \(P_{k}^{-}\) first.
\begin{lemma} \label{lem:ub of Pminus}
Let Assumption \ref{assum: initial condition} holds. Then, \(P_{\infty}\) is a uniform upper bound to \(P_{k}^{-}\), i.e., \(P_k^- \preccurlyeq P_{\infty}, \, \forall k \geq 0\).
\end{lemma}  
\begin{proof}
From \eqref{eq:kalman filter_initialization}, $P_0^- = \Sigma_0$.
Consider the Kalman filter whose initial covariance matrix is set to $P_{\infty}$ and gain is constant, i.e., \(L = P_{\infty}C^T(CP_{\infty}C^T+R)^{-1}\).
Since the \(L\) and \(P_{\infty}\) are the steady-state solution of the Kalman filter, the values of \(L\) and \(P_{\infty}\) remain constant over time.
Then, by using Theorem 2.1 from \cite{caines1970discrete}, if \(P_0^{-} \preccurlyeq P_{\infty}\), then \(P_{k}^{-} \preccurlyeq P_{\infty}, \, \forall k \geq 0\).
\end{proof}
Note that, from the Theorem 2.2 in \cite{caines1970discrete}, \(P_{k}^- \to P_{\infty}\) as \(k \to \infty\).
Based on the Lemma \ref{lem:ub of Pminus}, we derive \(P_\mathrm{MD}\).
\begin{lemma} \label{lem:ub of Pplus}
(Uniform upper bound to $P_k $)
Let Assumption \ref{assum: initial condition}-\ref{assum: system matrix} holds.
Then, we have that
\begin{equation} \label{eq:ub of P_plus}
\begin{split}
    & P_{k}  \preccurlyeq P_\mathrm{MD}, \, \forall k \geq 0,\\
    & P_\mathrm{MD}  = A^{-1}(P_{\infty}-Q)A^{-\top}.
\end{split}
\end{equation}
\end{lemma}  
\begin{proof}
(\ref{eq:ub of P_plus}) can be proved by contradiction.
Assume  \(\exists k \geq 0, \, P_\mathrm{MD}  \prec P_{k} \).
Then, from the Kalman filter covariance dynamics \cite{simon2006optimal}, $P_{k+1}^{-} = A P_k  A^\top + Q$.
By the assumption, $A P_k  A^\top + Q \succ A P  A^\top + Q = P_{\infty}$.
This means $P_{k+1}^{-} \succ P_{\infty}$ which contradicts \textit{Lemma} \ref{lem:ub of Pminus}.  
Therefore, $P_{k}  \preccurlyeq P_\mathrm{MD} , \forall k \geq 0$.
\end{proof}

Next we calculate an uniform upper bound to \(\Phi_k \) during the task, i.e. \(\Phi_k  \preccurlyeq \Phi_\mathrm{MD}, \forall k \geq 0\).
\begin{lemma} \label{lem:ub of phi}
(Uniform upper bound to $\Phi_k$)
Let Assumption \ref{assum: initial condition} hold.  
Then, we have that
\begin{equation} \label{eq:ub of phi_k}
\begin{split}
    & \Phi_k \preccurlyeq P_{\infty} = \Phi_\mathrm{MD}, \forall k \geq 0.\\
\end{split}
\end{equation}
\end{lemma}  
\begin{proof}
From the Kalman filter covariance update equation in \cite{simon2006optimal}, $P_k  = P_{k}^{-} - \Phi_k, \forall k \geq 0$.
Since  $P_k $ is positive semi-definite by definition of covariance matrix, we have $P_{k}^{-} - \Phi_k \succcurlyeq 0, \forall k \geq 0$.
By the lemma \ref{lem:ub of Pminus}, $P_{\infty} \succcurlyeq \Phi_k , \forall k \geq 0$.
\end{proof}
Note that the upper bound in lemma \ref{lem:ub of phi} is conservative because we regard $P_k $ as a zero matrix.
\begin{remark}
The upper bounds $P_\mathrm{MD}$ \eqref{eq:ub of P_plus} and $\Phi_\mathrm{MD}$ \eqref{eq:ub of phi_k} have advantages compared to the upper bounds $P_\mathrm{LJ}$ and $\Phi_\mathrm{LJ}$.
As these bounds are upper bounds to $P_k$ and $\Phi_k$ for all $k \geq 0$, we can use them for infinite time optimal control problem.
On the other hands, there exist disadvantages of the upper bounds $P_\mathrm{MD}$ \eqref{eq:ub of P_plus} and $\Phi_\mathrm{MD}$ \eqref{eq:ub of phi_k}.
As the assumptions \ref{assum: initial condition}-\ref{assum: system matrix} must hold for these upper bounds $P_\mathrm{MD}$ and $\Phi_\mathrm{MD}$, these bounds are valid for limited cases. 
Moreover, by derivation, $\Phi_\mathrm{MD}$ \eqref{eq:ub of phi_k} is conservative.
\end{remark}

\subsection{Uniformly Bounded Confidence Sets}
\label{sub:Uniformly Bounded Confidece Set}
Recall \textit{Definition}~\ref{def:uniform confidence set}. 
In the following, we show a way to design a uniformly bounded confidence set. 
\begin{lemma} \label{lem: uniformly bounded confidence set}
Let \(r_k  \in \mathbb{R}^{n_r} \sim \mathcal{N}(0, \Sigma_k )\), and \(\Sigma\)  a uniform upper bound of \(\Sigma_k \).
Define
\begin{equation} \label{eq:uniform confidence polytope}
\begin{split}
    & \mathcal{E}_{1-p}^{r} = \{r_k  : H_r r_k  \leq \bar{h}_{r}, k = 0,...,T-1\}, \\
    & [\bar{h}_{r}]_m = \mathrm{cdf}^{-1}(1-p_{m}) \sqrt{[H_r]_m \Sigma [H_r]_m^\top},\\
    & {\scriptstyle\sum}_{m=1}^{n_r} p_{m} = p, \, p_{m} \in (0,1), \\
    & m = 1,...,n_r, 
\end{split}
\end{equation}
where $H_r$ is a design parameter and \(\mathrm{cdf}^{-1} (\cdot)\) is the inverse of the cumulative Gaussian distribution function.
Then, \(\mathcal{E}_{1-p}^{r}\) is a uniformly bounded confidence set of probability level \(1-p\) for the random variable \(r_k \).
\end{lemma}
\begin{proof}
See Appendix. 
\end{proof}
Based on \textit{Lemma} \ref{lem: uniformly bounded confidence set}, we construct the following two \textit{Uniformly Bounded Confidence Sets} using the uniform upper bounds \(P\) and \(\Phi\):
\begin{itemize}
    \item \(\mathcal{E}_{1-p_x}^e\) : the uniformly bounded confidence set of probability level \(1-p_x\) for the random variable \(e_k\).
    \item \(\mathcal{E}_{1-p_f}^n\) : the uniformly bounded confidence set of probability level \(1-p_f\) for the random variable \(n_k\).
\end{itemize}
In this manuscript, each set is designed to be a compact polytope and contain the origin in its interior. 
The design parameter $H_r$ of each set is constructed using the right eigenvectors of the corresponding matrix describing the upper bound.
For example, in the case of \textit{Lemma} \ref{lem: uniformly bounded confidence set}, we use an eigenvector decomposition of the upper bound matrix $\Sigma$ as $\Sigma V = V D$, where $V$ is a matrix whose columns are right eigenvectors of $\Sigma$. Then, we devise $H_r=\begin{bmatrix} V & -V \end{bmatrix}^\top$ and construct \textit{Uniformly Bounded Confidence Set} as in \textit{Lemma} \ref{lem: uniformly bounded confidence set}. 
In the 2D case, this \textit{Uniformly Bounded Confidence Set} is the tight bounding rectangle of the ellipse $x^\top {\Sigma}^{-1} x \leq \mathrm{cdf}^{-1}(1-p_{m})$ which aligns with the ellipse's principal axes.

\section{Constraint tightening}
In this section, we describe a constraint tightening strategy based on the \emph{shrinking tube} structure \cite{chisci2001systems} using \textit{Uniformly Bounded Confidence Sets} $\mathcal{E}_{1-p_x}^e$ and $\mathcal{E}_{1-p_f}^n$. 
We prove that this strategy allows us to guarantee a quantifiable probability of success, which is the second contribution of this paper.
The underlying idea is to guarantee recursive feasibility of the MPC problem in \eqref{eq:MPC} with a user specified bound on its probability.
We refer to this property of the MPC problem \eqref{eq:MPC} as \emph{probabilistic recursive feasibility}.

\subsection{State constraint tightening}
Tightening the state constraints consists of two sequential steps.
First, we tighten the given state set \(\mathcal{X}\) to satisfy the chance constraint.
We can derive \(\hat{\mathcal{X}}\) using the uniformly bounded confidence set \(\mathcal{E}_{1-p_x}^e\):
\begin{equation} \label{eq:state set bar explicit}
\begin{split}
    & \hat{\mathcal{X}} = \mathcal{X} \ominus \mathcal{E}^{e}_{1-p_x}.
\end{split}
\end{equation}
This make sure that if the estimation error $e_k$ is in the set $\mathcal{E}^{e}_{1-p_x}$ and $\hat{x}_k  \in \hat{\mathcal{X}}$, then the true state $x_k$ satisfies the state constraint $\mathcal{X}$.
By construction of $\mathcal{E}^{e}_{1-p_x}$ in \ref{sub:Uniformly Bounded Confidece Set}, the set $\hat{x}_k  \in \hat{\mathcal{X}}$ implies the satisfaction of the state chance constraint in \eqref{eq:constraints}, i.e., $\hat{x}_k  \in \hat{\mathcal{X}} \Rightarrow \mathbb{P}(x_k \in \mathcal{X}| \hat{x}_k ) \geq 1-p_x$.
\begin{remark}
The nominal state $\bar{x}_{0|k}$ is initialized with the estimated state $\hat{x}_k \in \hat{\mathcal{X}}$ in \eqref{eq:MPC}.
As mentioned before, due to unboundedness of noise and disturbance supports and boundedness of inputs, the MPC problem \eqref{eq:MPC} might become infeasible.
Thus, we aim to regulate the rate of loss of recursive feasibility below a user-specified bound $p_f$ by further tightening $\hat{\mathcal{X}}$ using $\mathcal{E}_{1-p_f}^n$.
\end{remark}

Second, we tighten \( \hat{\mathcal{X}}\) using the shrinking tube approach in  \cite{chisci2001systems} and impose the tightened constraints on the nominal states to ensure recursive feasibility of the MPC problem with probability no smaller than \(1-p_f\), i.e., regulate the loss of recursive feasibility rate to a value no larger than $p_f$. The following equations are the direct result of the shrinking tube tightening applied to (\ref{eq:control error dynamics}) with the uncertain set $\mathcal{E}_{1-p_f}^n$:
\begin{equation} \label{defition of the state set}
\begin{split}
    & \Bar{\mathcal{X}}_{\mathrm{RF}, 0} =  \hat{\mathcal{X}}, \\
    & \Bar{\mathcal{X}}_{\mathrm{RF}, i} =  \hat{\mathcal{X}} \ominus \bigoplus_{q=0}^{i-1} {A^q_\mathrm{cl}} \mathcal{E}_{1-p_f}^n, \\
    & i=1,2,...,N-1. 
\end{split}
\end{equation}

\subsection{Control input constraint tightening}
Similar to \eqref{defition of the state set}, we tighten the control input constraints using the shrinking tube approch in \cite{chisci2001systems} to ensure recursive feasibility  with probability no smaller than \(1-p_f\). We obtain:
\begin{equation} \label{defition of the control set}
\begin{split}
    & \bar{\mathcal{U}}_{\mathrm{RF}, 0} = \mathcal{U}, \\
    & \bar{\mathcal{U}}_{\mathrm{RF}, i} = \mathcal{U} \ominus \bigoplus_{q=0}^{i-1} K{A^q_\mathrm{cl}}  \mathcal{E}_{1-p_f}^n, \\
    & i=1,2,...,N-1. 
\end{split}
\end{equation}

\subsection{Terminal set construction}
The terminal set is constructed into two sequential steps.
First, we construct a robust positive invariant set \(\hat{\mathcal{X}}_{f}\) for the uncertain system \eqref{eq:estimated state dynamics} under the linear control  policy $u_k = K \hat{x}_k$.
Robust positive invariant sets are defined as follows \cite{borrelli2017predictive}.
\begin{definition} \label{def:RPIS}
(Robust Positive Invariant Set) A set $\mathcal{O} \subseteq \mathcal{X}$ is said to be a robust positive invariant set for the autonomous system $x_{k+1}=Ax_{k} + w_k$, where $w_k \in \mathcal{W}$, if $x_k \in \mathcal{O} \Rightarrow x_{k+1} \in \mathcal{O}, \, \forall w_k \in \mathcal{W}$.
\end{definition}

Second, to regulate the loss of feasibility rate, we further tighten \(\hat{\mathcal{X}}_{f}\) as follows:
\begin{equation} \label{defition of the t set}
\begin{split}
    & \bar{\mathcal{X}}_{f,\mathrm{RF}} = \hat{\mathcal{X}}_{f} \ominus \bigoplus_{q=0}^{N-1} {A^q_\mathrm{cl}} \mathcal{E}_{1-p_f}^n. 
\end{split}
\end{equation}

\section{Controller Properties}

\begin{theorem} \label{thm:Probabilistic recursive feasibility}
(Probabilistic recursive feasibility)
Consider the closed-loop system \eqref{eq:system} when the control input \eqref{eq:MPC policy} is obtained by solving the MPC \eqref{eq:MPC} where the nominal state $\bar{x}_{0|k}$ is initialized with the state estimate $\hat{x}_k$ obtained from the Kalman filter \eqref{eq:kalman filter} and the constraints are computed as in \eqref{defition of the state set}, \eqref{defition of the control set}, and \eqref{defition of the t set}.
For any $k \geq 0 $, if the MPC \eqref{eq:MPC} is feasible at time step $k$, then with probability no smaller than \(1-p_f\), the MPC \eqref{eq:MPC} is feasible at time step $k+1$.
\end{theorem}
\begin{proof}
Let the optimal auxiliary input sequence from time step $k$ as \(c^\star_{0:N-1|k}\) and the corresponding predicted nominal state as $\bar{x}_{0:N|k}$.
With probability no smaller than \(1-p_f\), \(n_k  \in \mathcal{E}_{1-p_f}^n\) by definition of the \textit{Uniformly Bounded Confidence Set} in \textit{Definition} \ref{def:uniform confidence set}.
Thus, for proving \textit{Theorem} \ref{thm:Probabilistic recursive feasibility}, it is sufficient to prove recursive feasibility of the MPC \eqref{eq:MPC} when the additive disturbance at time step $k$ is in $\mathcal{E}_{1-p_f}^n$, i.e., \(n_k  \in \mathcal{E}_{1-p_f}^n\).

From $\bar{x}_{0|k}  = \hat{x}_k $ and \eqref{eq:control error dynamics}, the estimated state $\hat{x}_{k+1} $ is re-written as:
\begin{equation} \label{eq:next state estimated}
\begin{split}
    & \hat{x}_{k+1}  = \bar{x}_{1|k} + \Tilde{n}_k, 
\end{split}
\end{equation}
where \(\Tilde{n}_k\in \mathcal{E}_{1-p_f}^n\) is the realized disturbance in \eqref{eq:estimated state dynamics} at time step \(k\). We consider a candidate solution to the MPC \eqref{eq:MPC} at time step \(k+1\) based on the optimal solution to the MPC \eqref{eq:MPC} at time step \(k\) as follows:
\begin{equation} \label{eq:candidate solution j=1}
\begin{split}
    &  \Tilde{c}_{i|k+1} =  \begin{cases} 
     c^\star_{i+1|k} + K A_\mathrm{cl}^{i} \Tilde{n}_k, & i = 0,...,N-2, \\
    K (\bar{x}_{N|k} + A_\mathrm{cl}^{i} \Tilde{n}_k), & i=N-1. 
   \end{cases} 
\end{split}
\end{equation}
Under this candidate solution and \eqref{eq:next state estimated}, the states $\bar{x}_{0:N|k+1}$ at time step $k+1$ are derived as:
\begin{equation} \label{eq: predicted states}
\begin{split}
    & \bar{x}_{i|k+1} =  \begin{cases} 
    \bar{x}_{i+1|k} + A_\mathrm{cl}^{i} \Tilde{n}_k, & i = 0,...,N-1,\\
     A_\mathrm{cl} \bar{x}_{N-1|k+1} , & i=N.
   \end{cases} 
\end{split}
\end{equation}
Following \cite[Lemma~7]{chisci2001systems}, the candidate solution in \eqref{eq:candidate solution j=1} and the resulting state in \eqref{eq: predicted states} satisfy the control constraints and the state constraints, respectively, i.e., $\Tilde{c}_{i|k+1}\in \bar{\mathcal{U}}_{\mathrm{RF}, i}, \, \bar{x}_{i|k+1} \in \Bar{\mathcal{X}}_{\mathrm{RF}, i}\, \forall i\in\{0,1,...,N-2\}$.


Now we show the candidate solution when $i = N-1$ satisfies the control input constraint in \eqref{defition of the control set}.
\begin{equation} \label{eq:last candidate solution proof}
\begin{split}
    & \quad \,\,\, \bar{x}_{N|k} \in \bar{\mathcal{X}}_{f,\mathrm{RF}}\\
    & \Rightarrow \bar{x}_{N|k} + {\scriptstyle\sum}_{q=0}^{N-1} A_\mathrm{cl}^q n_q \in \hat{\mathcal{X}}_{f} , \, \forall n_q \in \mathcal{E}_{1-p_f}^n (\because \eqref{defition of the t set})\\
    & \Rightarrow K(\bar{x}_{N|k} + {\scriptstyle\sum}_{q=0}^{N-1} A_\mathrm{cl}^q n_q) \in \mathcal{U}, \, \forall n_q \in \mathcal{E}_{1-p_f}^n \\
    & \Rightarrow \underbrace{K(\bar{x}_{N|k}  + A_\mathrm{cl}^{N-1} \Tilde{n}_k )}_{\Tilde{c}_{N-1|k+1} \, (\because \eqref{eq:candidate solution j=1})} \in \underbrace{\mathcal{U} \ominus \bigoplus_{q=0}^{N-2} K{A^q_\mathrm{cl}} \mathcal{E}_{1-p_f}^n}_{\bar{\mathcal{U}}_{\mathrm{RF}, N-1} (\because \eqref{defition of the control set})}. \\
    & \, \therefore \Tilde{c}_{N-1|k+1} \in \bar{\mathcal{U}}_{\mathrm{RF}, N-1}. 
\end{split}
\end{equation}
The third line of \eqref{eq:last candidate solution proof} results from the fact that $\hat{\mathcal{X}}_{f}$ is a robust positive invariant set.


Now we prove that satisfying the terminal constraints at the time step \(k\), $\bar{x}_{N|k}  \in \bar{\mathcal{X}}_{f,\mathrm{RF}}$, implies the satisfaction of the state constraints at the time step \(k+1\), $\bar{x}_{N-1|k+1}  \in \Bar{\mathcal{X}}_{\mathrm{RF}, N-1}$, as follows:
\begin{equation} \label{eq:last state proof}
\begin{split}
    & \quad \,\,\, \bar{x}_{N|k}  \in \bar{\mathcal{X}}_{f,\mathrm{RF}} \\
    & \Rightarrow \bar{x}_{N|k}  + {\scriptstyle\sum}_{q=0}^{N-1} A_\mathrm{cl}^q n_q \in \hat{\mathcal{X}}_{f} , \, \forall n_q \in \mathcal{E}_{1-p_f}^n (\because \eqref{defition of the t set})\\
    & \Rightarrow \bar{x}_{N|k}  + A_\mathrm{cl}^{N-1} \Tilde{n}_k  \in \hat{\mathcal{X}}_{f} \ominus \bigoplus_{q=0}^{N-2} {A^q_\mathrm{cl}} \mathcal{E}_{1-p_f}^n \\
    & \Rightarrow \underbrace{\bar{x}_{N|k}  + A_\mathrm{cl}^{N-1} \Tilde{n}_k}_{\bar{x}_{N-1|k+1} (\because \eqref{eq: predicted states})} \in \underbrace{\hat{\mathcal{X}} \ominus \bigoplus_{q=0}^{N-2} {A^q_\mathrm{cl}} \mathcal{E}_{1-p_f}^n}_{ \Bar{\mathcal{X}}_{\mathrm{RF}, N-1} (\because \eqref{defition of the state set})}. \\
    & \therefore \bar{x}_{N-1|k+1}  \in \Bar{\mathcal{X}}_{\mathrm{RF}, N-1}.  
\end{split}
\end{equation}
The fourth line of \eqref{eq:last state proof} results from the fact that $\hat{\mathcal{X}}_{f} \subseteq \hat{\mathcal{X}}$.
Therefore, the states at time step $k+1$ satisfies the state constraints, i.e., $\bar{x}_{i|k+1}  \in \Bar{\mathcal{X}}_{\mathrm{RF}, i}, \, i=0,1,...,N-1$. 

Finally, we show that $\bar{x}_{N|k+1}  \in \bar{\mathcal{X}}_{f,\mathrm{RF}} $.
Note that since $\bar{\mathcal{X}}_{f,\mathrm{RF}} $ is not a robust invariant set, this claim is not trivial.
We prove the claim as follow:
\begin{equation} \label{eq:terminal set proof}
\begin{split}
    & \quad \,\,\, \bar{x}_{N|k}  \in \bar{\mathcal{X}}_{f,\mathrm{RF}} \\
    & \Rightarrow \bar{x}_{N|k}  + {\scriptstyle\sum}_{q=0}^{N-1} A_\mathrm{cl}^q n_q \in \hat{\mathcal{X}}_{f} , \, \forall n_q \in \mathcal{E}_{1-p_f}^n (\because \eqref{defition of the t set})\\
    & \Rightarrow A_\mathrm{cl} \bar{x}_{N|k}  + {\scriptstyle\sum}_{q=0}^{N} A_\mathrm{cl}^q n_q  \in \hat{\mathcal{X}}_{f}, \, \forall n_q \in \mathcal{E}_{1-p_f}^n\\
    & \Rightarrow \underbrace{A_\mathrm{cl} (\bar{x}_{N|k}  + A_\mathrm{cl}^{N-1} \Tilde{n}_k )}_{\bar{x}_{N|k+1} \, (\because \eqref{eq: predicted states})} \in  \underbrace{\hat{\mathcal{X}}_f \ominus \bigoplus_{q=0}^{N-1} {A_\mathrm{cl}}^{q} \mathcal{E}_{1-p_f}^n}_{\bar{\mathcal{X}}_{f,\mathrm{RF}} (\because \eqref{defition of the t set})}. \\
     & \therefore \bar{x}_{N|k+1}  \in \bar{\mathcal{X}}_{f,\mathrm{RF}}.  
\end{split}
\end{equation}
The third line of \eqref{eq:terminal set proof} results from the fact that $\hat{\mathcal{X}}_{f}$ is a robust positive invariant set.
Thus, the candidate solution \eqref{eq:candidate solution j=1} is the feasible solution of the MPC \eqref{eq:MPC} at time step $k+1$.
\end{proof}

\begin{corollary} \label{cor:Constraint satisfaction}
(Chance Constraint Satisfaction)
Consider the closed-loop system \eqref{eq:system} when the control input \eqref{eq:MPC policy} is obtained by solving the MPC \eqref{eq:MPC} where the nominal state $\bar{x}_{0|k}$ is initialized with the state estimate $\hat{x}_k$ obtained from the Kalman filter \eqref{eq:kalman filter} and the constraints are computed as \eqref{defition of the state set}, \eqref{defition of the control set}, and \eqref{defition of the t set}.
Let the MPC in \eqref{eq:MPC} be feasible at time step $k$. 
Then, the constraints \eqref{eq:constraints} are satisfied, i.e.,
\begin{equation} \label{eq:constraints recall}
\begin{split}
    & \qquad \, \, \mathbb{P}(x_{k} \in \mathcal{X}) \geq 1-p_x,\\
    & \qquad \, \, u_k=\pi_k(\hat{x}_k ) \in \mathcal{U}.
\end{split}
\end{equation}
\end{corollary}
\begin{proof}
We first prove $\mathbb{P}(x_{k} \in \mathcal{X}) \geq 1-p_x$ if the MPC in \eqref{eq:MPC} is feasible given the initial condition $\hat{x}_k $ at time step $k$. As $\bar{x}_{0|k} = \hat{x}_k $, the following holds, 
\begin{equation} \label{eq:estimated state in xbar}
\begin{split}
    & \bar{x}_{0|k} \in \bar{\mathcal{X}}_{\mathrm{RF},0} \iff \hat{x}_{k}  \in \bar{\mathcal{X}}_{\mathrm{RF},0} =  \hat{\mathcal{X}}. 
\end{split}
\end{equation}
By \eqref{eq:state set bar explicit}, $ \hat{\mathcal{X}} \oplus \mathcal{E}^{e}_{1-p_x} \subseteq \mathcal{X}$.
Moreover, as per the definition of \(\mathcal{E}^{e}_{1-p_x}\), \(\mathbb{P}(x_{k} - \hat{x}_{k} \in \mathcal{E}^{e}_{1-p_x}) \geq 1-p_x, \, \forall k \geq 0\).
As \(\hat{x}_{k} \in  \hat{\mathcal{X}}\) in \eqref{eq:estimated state in xbar}, \(\mathbb{P}(x_{k} \in \mathcal{X}) \geq 1-p_x\).

Second, we prove $\pi_k(\hat{x}_{i|k}) \in \mathcal{U}$.
As the MPC in \eqref{eq:MPC} is feasible, $c^\star_{0|k} \in \bar{\mathcal{U}}_{\mathrm{RF},0}=\mathcal{U}$.
In closed-loop, we will apply \(\pi_k(\hat{x}_k ) = c^\star_{0|k}\) at the time step \(k\).
Thus, $u_k=\pi_k(\hat{x}_k ) \in \mathcal{U}$ is satisfied. 
\end{proof}

\section{Simulation Results}
In this section, we find approximate solutions to the following optimal control problem in receding horizon fashion with horizon $N=5$.
\begin{equation} \label{eq:ftocp example}
\begin{split}
    & V (\mu_0, \mathcal{F}) = \\
    & \min_{U (\cdot)} \,\, \sum_{k=0}^{T-1} \hat{x}_k^\top \begin{bmatrix} 100 & 0\\ 0 & 1 \end{bmatrix} \hat{x}_k + \pi_k (\hat{x}_k)^2  \\
    & \,\,\, \textnormal{s.t.,} \,\,\,  x_{k+1}  = A x_k  + B \pi_k (\hat{x}_k) + w_k ,\\
    & \qquad \, \, y_k  = C x_k  + v_k ,  \\
    & \qquad \, \, \hat{x}_{k+1} = f_k(\hat{x}_{k}, y_{k+1}), \\
    & \qquad \, \, \hat{x}_{0} = \mu_0 + L_0 (y_0 - C \mu_0), \\
    & \qquad \, \, \mathbb{P}(\begin{bmatrix} -8 \\ -8 \end{bmatrix} \leq x_{k}  \leq \begin{bmatrix} 80 \\ 40 \end{bmatrix}) \geq 0.95,\\
    & \qquad \, \, |\pi_k (\hat{x}_k)| \leq 5,\\
    & \qquad \, \, x_0  \sim \mathcal{N}(\mu_0 , 0.1 \mathrm{I}), \\ 
    & \qquad \, \, w_k \sim \mathcal{N}(0,0.1 \mathrm{I}), \,\, v_k \sim \mathcal{N}(0,0.1),\\
    & \qquad \, \, k=0,...,T-1, 
\end{split}
\end{equation}
where $T=50$, $A=\begin{bmatrix} 1 & 1\\ 0 & 1 \end{bmatrix}, B=\begin{bmatrix} 0.5 \\ 1 \end{bmatrix}$, $C=\begin{bmatrix} 1 & 0 \end{bmatrix}$, and $\mu_0= \begin{bmatrix} 25 & 0 \end{bmatrix}^\top$. 
For the proposed MPC, with parameters $Q_\mathrm{LQR}=\begin{bmatrix} 100 & 0\\ 0 & 1 \end{bmatrix}$ and $R_\mathrm{LQR} = 1$, the control gain $K$ is chosen to be the optimal LQR gain and the terminal cost is designed $O(x) = x^\top P_\mathrm{LQR} x$ where $P_\mathrm{LQR}$ is a solution of the discrete time algebraic Riccati equation.

We compare the proposed algorithm with two algorithms. 
First, we simply feed an estimated state from the Kalman filter as the true state to a deterministic MPC problem. The formulation of this problem is given in the Appendix. We will refer to this first algorithm as the \emph{baseline}. 
The second algorithm for comparisons is the SMPC-c algorithm in \cite{farina2015approach}. The SMPC-c algorithm in \cite{farina2015approach} is designed to use the Luenberger observer and satisfy chance input constraints, which allow the controller to exert potentially unbounded control inputs.
For the SMPC-c algorithm \cite{farina2015approach}, we set the Luenberger observer with constant observer gain $L_{\infty} =L_k, \forall k \geq 0$ and the input chance constraint as $\mathbb{P}(|\pi_k (\hat{x}_k)| \leq 5) \geq 0.99$ to comply with \cite{farina2015approach}. Note that the constant observer gain $L_{\infty}$ is set to the steady-state Kalman filter gain \cite{simon2006optimal}.

We run 10,000 Monte-Carlo simulations of the system in \eqref{eq:ftocp example} in closed-loop using each algorithm  to investigate task failure rate and constraint violation rate of each algorithm. The results are presented in Table. \ref{table:statistics 2}.
As shown in the Table \ref{table:statistics 2}, the task failure rate of the SMPC-c \cite{farina2015approach} is high compared to the proposed control algorithm by 2000 times.
This is because SMPC-c \cite{farina2015approach} is designed when the systems are not imposed on hard control constraints, but our system is.
The proposed algorithm is designed to regulate the task failure rate when the system is imposed on hard control input constraints, and the results in Table \ref{table:statistics 2}  show that it can regulate the task failure lower than the theoretical value $1-(1-p_f)^{T-1}$.
\begin{table}[h]
\centering
\caption{Task failure rate}
\label{table:statistics 2}
\begin{tabular}{ ||l||c||c||} 
 \hline
Algorithms & Numerical value & $1-(1-p_f)^{T-1}$ \\
 \hline
 Baseline & 0.7284 & - \\
 SMPC-c \cite{farina2015approach} & 0.1731 & - \\ 
 Proposed & $8 \times10^{-4}$ & 0.095 \\ 
 \hline
\end{tabular}
\end{table}

Constraint violation in one closed-loop simulation at time $k$ is defined as  \(x_k \not\in \mathcal{X}\). The desired probability of constraint violation is set in (\ref{eq:ftocp example}) to be lower than \(p_x=0.05\).
We calculate constraint violation rate as dividing the number of the constraint violations in all successful simulations by the number of time steps of successful simulations.
The results are in Table \ref{table:statistics 1}.
For all cases, the violation rate is much lower than the theoretical upper bound because of the convervativeness of the constraint tightening.
\begin{table}[h]
\centering
\caption{Constraint violation rate \(p_x\)}
\label{table:statistics 1}
\begin{tabular}{ ||l||c||c||} 
 \hline
 Algorithms & Numerical value & $p_x$\\
 \hline
 Baseline & 0.0165 & 0.05 \\
 SMPC-c \cite{farina2015approach} & $1.96 \times10^{-4}$ & 0.05 \\ 
 Proposed & $4.00 \times10^{-6}$ & 0.05 \\ 
 \hline
\end{tabular}
\end{table}

\section{CONCLUSIONS}
We proposed an output feedback SMPC for constrained uncertain LTI systems subject to chance constraints on states and hard bounded constraints on inputs. 
The uncertainties included additive disturbance on state and additive measurement noise on input which are Gaussian distributed.
We propose a way to construct rigid \textit{Uniformly Bounded Confidence Sets} and an output feedback SMPC is designed by tightening constraints using these set. 
We prove that the proposed SMPC guaranteed the task completion with no smaller than quantifiable probability while satisfying the imposed chance constraints on states and hard bounded constraints on inputs.
The effectiveness of the proposed output feedback SMPC was investigated via a numerical example and compared with classical output feedback SMPC.

\renewcommand{\baselinestretch}{0.8}
\printbibliography 
\balance
\section*{Appendix}
\subsection{L\"{o}wner-John ellipsoid problem}
We solve the following optimization problem to numerically calculate the uniform upper bound $\psi$ given $\psi_{0:T-1}$:
\begin{equation} \label{lowner john ellipsoid}
\begin{split}
    & \min_{\psi, \tau_0,..., \tau_{T-1}} \mathrm{log~det}~ \Phi\\
    & \qquad \textnormal{s.t.,} \quad \tau_k \geq 0, \\
    & \qquad \qquad \, \, \begin{bmatrix} \psi^{-2} -\tau_k \psi_k^{-1} & 0 \\
    0 & -1+\tau_k\\
    \end{bmatrix} \preccurlyeq 0,\\
    & \qquad \qquad \, \, k = 0,...,T-1.
\end{split}
\end{equation}
The optimal solution for this optimization problem (\ref{lowner john ellipsoid}) \(\psi\) is the upper bound of the error covariance matrix \(\psi_k, \, k = 0,...,T-1\) over the task. 

\subsection{Proof of Lemma \ref{lem: uniformly bounded confidence set}}
\label{Proof of lemma 1}
\begin{proof}
Consider the following joint chance constraint:
\begin{equation} \label{eq:joint chance constraint}
\begin{split}
    & \mathcal{R}_k  = \{r_k : H_r r_k  \leq {h}_{r_k }\}, \\
    & \mathbb{P}(r_k  \notin \mathcal{R}_k ) < p.
\end{split}
\end{equation}
Note that this constraint can also be written as \(\mathbb{P}(H_r r_k  \leq {h}_{r_k }) \geq 1-p\).
By Boole's inequality, the joint chance constraint in (\ref{eq:joint chance constraint}) can be split into multiple individual chance constraints as follows:
\begin{equation} \label{eq:single chance constraint}
\begin{split}
    & \mathbb{P}([H_r]_m r_k  > [{h}_{r_k }]_m) \leq p_m, \\
    & {\scriptstyle\sum}_{m=1}^{n_r} p_{m} = p, \, p_{m} \in (0,1), \\
    & m = 1,...,n_r.
\end{split}
\end{equation}
By Proposition 2 of \cite{farina2015approach}, each constraint in (\ref{eq:single chance constraint}) is verified if
\begin{equation} \label{eq:single chance constraint sufficient condition}
\begin{split}
    & \mathrm{cdf}^{-1}(1-p_{m}) \sqrt{[H_r]_m \Sigma_k  [H_r]_m^\top} < [{h}_{r_k }]_m,\\
    &  p_{m} \in (0,1), \, m = 1,...,n_r.
\end{split}
\end{equation}
Using this fact, the following set satisfies the joint chance constraint in (\ref{eq:joint chance constraint}):
\begin{equation} \label{eq:confidence polytope}
\begin{split}
    & \mathcal{E}_{1-p}^{r_k } = \{r_k  : H_r r_k  \leq \bar{h}_{r_k }\}, \\
    & [\bar{h}_{r_k }]_m = \mathrm{cdf}^{-1}(1-p_{m}) \sqrt{[H_r]_m \Sigma_k  [H_r]_m^\top},\\
    & {\scriptstyle\sum}_{m=1}^{n_r} p_{m} = p, \, p_{m} \in (0,1), \\
    & m = 1,...,n_r.
\end{split}
\end{equation}
Considering \(\Sigma_k  \preccurlyeq \Sigma\), the following inequality is true: \([\bar{h}_{r_k }]_m \leq [\bar{h}_{r}]_m\).
 Thus, \(\mathcal{E}_{1-p}^{r_k } \subseteq \mathcal{E}_{1-p}^{r}, \forall k \in \{0,...,T-1\}\).
This means that \(\mathcal{E}^{r}_{1-p}\) satisfies \(\mathbb{P}(r_k  \in \mathcal{E}^{r}_{1-p}) \geq 1-p\).
 Therefore, the set \(\mathcal{E}_{1-p}^{r}\) satisfies the \textit{Definition} \ref{def:uniform confidence set}, and thus is a uniformly bounded confidence set.
\end{proof}

\subsection{Baseline MPC problem} 
The baseline algorithm consists of two consecutive parts. 
First, given the measurement, the Kalman filter \cite{simon2006optimal} in (\ref{eq:kalman filter}) and (\ref{eq:kalman filter_initialization}) estimates actual state. 
Second, a baseline MPC in (\ref{eq:baseline}) calculates a control input assuming the estimated state as the true state and treating the system as a deterministic system. The MPC problem is:
\begin{equation} \label{eq:baseline}
\begin{split}
    & {V}^{\mathrm{BL}}_{k \rightarrow k+N}(\hat{x}_k ) = \\
    & \min_{\substack{c _{0|k},...c _{N-1|k}}} \sum_{i=0}^{N-1} \ell({x}_{i|k}, c_{i|k}) + O({x}_{N|k})\\
    & ~~~~~ \textnormal{s.t.,} \quad ~~ {x}_{i+1|k} = A {x}_{i|k} + B c_{i|k}, \\
    & ~~~~~~~\qquad \, \, {x}_{0|k} = \hat{x}_k ,\\
    & ~~~~~~~\qquad \, \,  {x}_{i+1|k} \in  \mathcal{X}, \\
    & ~~~~~~~ \qquad \, \,  c_{i|k} \in \mathcal{U},\\
    & ~~~~~~~ \qquad \, \, {x}_{N|k} \in \mathcal{X}_{f},\\
    & ~~~~~~~ \qquad \, \, i=0,1,...,N-1,
\end{split}
\end{equation}
where $\mathcal{X}_{f}$ a positive invariant set for the nominal system in \eqref{eq:baseline} under the terminal control policy $\pi_f(x) = Kx$. We then apply $c^\star_{0|k}$ at time step $k$ to system \eqref{eq:system}.

\end{document}